\documentclass[runningheads]{llncs}
\usepackage[T1]{fontenc}
\usepackage{graphicx}
\usepackage{amsmath, amsthm, amssymb}
\usepackage[colorlinks=true, allcolors=blue]{hyperref}
\usepackage{cleveref}
\usepackage{xspace}
\usepackage{todonotes}
\usepackage{comment}
\usepackage{algpseudocode}
\usepackage{enumitem}
\usepackage{booktabs,tabularx}
\usepackage{lipsum}

\theoremstyle{plain}

\newtheorem{prop}[theorem]{Proposition}

\theoremstyle{definition}

\newtheorem{prb}{Problem}

\newcommand{\arity}{\mathsf{ar}}
\newcommand{\boldR}{\mathbf{R}}
\newcommand{\boldx}{\mathbf{x}}
\newcommand{\boldy}{\mathbf{y}}

\newcommand{\calH}{\mathcal{H}}
\newcommand{\calG}{\mathcal{G}}

\newcommand{\wpe}{\textsf{W[P]}\xspace}
\newcommand{\homo}{\textsc{Hom}$_<$\xspace}

\newcommand{\parhomo}{\textsc{Hom}$_<(|V(H)|)$\xspace}
\newcommand{\homoH}{\textsc{Hom}$_<$\textsc{(H)}\xspace}
\newcommand{\homos}{\textsc{Hom}$_<^\star$\xspace}
\newcommand{\NP}{\textsf{NP}\xspace}
\newcommand{\wone}{\textsf{W[1]}\xspace}
\newcommand{\XP}{\textsf{XP}\xspace}
\newcommand{\Oh}{\mathcal{O}}

\newcommand{\core}{\textsc{CORE}$_<$\xspace}

\newcommand{\PP}{\textsf{P}}

\newcommand{\homochi}{\textsc{Hom}$_<^{\chi^<(G)}$\xspace}

\newcommand{\problemStatement}[3]{%
  \begin{center}
  \begin{tabularx}{\columnwidth}{@{}lX@{}}
  \toprule
  \multicolumn{2}{@{}c@{}}{\textsc{#1}}\tabularnewline
  \midrule
  \bfseries Input:    & #2 \\
  \bfseries Question: & #3 \\
  \bottomrule
  \end{tabularx}
  \end{center}
}

\begin{document}
\title{Complexity {Aspects} of {Homomorphisms} of {Ordered} {Graphs}}
%
%
\author{Michal \v{C}ert\'{\i}k \inst{1}\orcidID{0009-0008-6880-4896} \and
Andreas Emil Feldmann \inst{2}\orcidID{0000-0001-6229-5332} \and
Jaroslav Ne\v{s}et\v{r}il \inst{1}\orcidID{0000-0002-5133-5586} \and
Pawe\l{} Rz\k{a}\.zewski \inst{3}\orcidID{0000-0001-7696-3848}\thanks{Supported by the National Science Centre grant 2024/54/E/ST6/00094.}}
\authorrunning{M. \v{C}ert\'{\i}k et al.}
%
\institute{Computer Science Institute, Faculty of Mathematics and Physics \\
Charles University\\
Prague, Czech Republic \and
Department of Computer Science \\
University of Sheffield\\
Sheffield, United Kingdom \and
Warsaw University of Technology\\
and University of Warsaw\\
Warsaw, Poland}
\maketitle              
\begin{abstract}

     We examine ordered graphs, defined as graphs with linearly ordered vertices, from the perspective of homomorphisms (and colorings) and their complexities. We demonstrate the corresponding computational and parameterized complexities, along with algorithms associated with related problems. These questions are interesting and we show that numerous problems lead to various complexities. The reduction from homomorphisms of unordered structures to homomorphisms of ordered graphs is proved, achieved with the use of ordered bipartite graphs. We then determine the \NP-completeness of the problem of finding ordered homomorphisms of ordered graphs and the \XP and \wone-hard nature of this problem parameterized by the number of vertices of the image ordered graph. Classes of ordered graphs for which this problem can be solved in polynomial time are also presented.

\end{abstract}

\section{Introduction}

An \emph{ordered graph} is a graph whose vertex set is totally ordered (see example in Figure~\ref{fig:OrdHomsInterval}).
For two ordered graphs $G$ and $H$, an \emph{ordered homomorphism} from $G$ to $H$ is a mapping $f$ from $V(G)$ to $V(H)$ that preserves edges and the orderings of vertices, i.e.,
\begin{enumerate}
    \item for every $uv \in E(G)$ we have $f(u)f(v) \in E(H)$,
    \item for $u,v \in V(G)$, if $u \leq v$, then $f(u) \leq f(v)$.
\end{enumerate}

Note that the second condition implies that the preimage of every vertex of $H$ forms a segment or an \emph{independent interval} in the ordering of $V(G)$ (see Figure~\ref{fig:OrdHomsInterval}).

We denote the existence of an ordered homomorphism from the ordered graph $G$ to the ordered graph $H$ by $G\to H$.

\begin{figure}[ht]
\begin{center}
\includegraphics[scale=0.8]{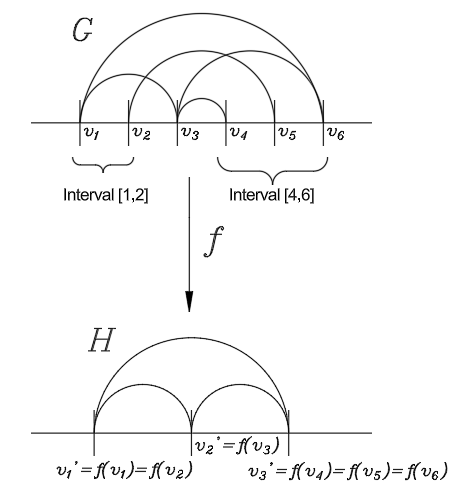}
\end{center}
\caption{Ordered Homomorphism $f$ and Independent Intervals.}
\label{fig:OrdHomsInterval}
\end{figure}


Further on we may call an independent interval simply an interval, and an ordered graph, and ordered homomorphism simply graph, and homomorphism, respectively.

\section{Motivation and Overview}

Ordered graphs frequently emerge in various contexts: extremal theory (~\cite{Pach2006},~\cite{conlon2016ordered}), category theory (~\cite{Hedrlín1967},~\cite{nesetril2016characterization}), Ramsey theory (~\cite{Nesetril1996},~\cite{Hedrlín1967},~\cite{balko2022offdiagonal}), among others. Recently, it has been shown that the concept of twin width in graphs corresponds to NIP ("not the independence property") classes of ordered graphs (~\cite{bonnet2021twinwidth}, see also ~\cite{bonnet2024twinwidth}), thereby linking graph theory with model theory.

The richness of the field of ordered graphs is reflected not only in its theoretical depth and challenges but also in its numerous applications across science and technology. Related research spans a wide range of domains, including physics \cite{verbytskyi2020hepmc3-80a}, medicine and biology \cite{goerttler2024machine-7eb}, large language models \cite{ge2024can-da2}, neural networks \cite{guo2019seq2dfunc-c64}, machine learning \cite{goerttler2024machine-7eb}, self-supervised learning \cite{LimOrderLearning2020}, data analysis and subspace clustering \cite{xing2025block-diagonal-073}, systems and networks \cite{li2015understanding-1cc}, software optimization \cite{romansky2020approach-60b}, malware detection \cite{thomas2023intelligent-c25}, business process management \cite{kourani2023business-759}, workflow models \cite{kourani2023scalable-a0f}, decision making \cite{wang2022improved-f6a}, dynamic system call sandboxing \cite{zhang2023building-6de}, fault tolerance \cite{chen2023pgs-bft-dcb}, blockchains \cite{malkhi2023bbca-chain-bfa}, curriculum development \cite{kuzmina2020curriculum-cdf}, multi-linear forms \cite{bhowmik2023multi-linear-2f5}, ordered graph grammars \cite{brandenburg2005graph-grammars-2fe}, rigidity theory \cite{connelly2024reconstruction-912}, shuffle squares \cite{grytczuk2025shuffle-a46}, and tilings \cite{balogh2022tilings-d8d}, among many others.

In relation to the aforementioned research, homomorphisms of ordered graphs provide both validation and extension: although they impose stricter conditions in comparison to standard homomorphisms (see, for example, ~\cite{HellNesetrilGraphHomomorphisms}), they also exhibit their own unique complexity (see, for instance~\cite{Axenovich2016ChromaticNO,Guerra2012,braun2013cellular-998,nie2023asymptotic-fa9,bose2004ordered-e8c,nescer2023duality}.

The exploration of complexities and parameterized complexities concerning ordered graphs and their homomorphisms has also been examined from multiple perspectives. Recently, ~\cite{kun2025dichotomy-dd6} has shown that ordering problems for graphs defined by finitely many forbidden ordered subgraphs still capture the class \NP. In ~\cite{duffus1995computational-7a8}, the complexities of decision problems involving ordered graphs and their subgraphs are studied. ~\cite{certik_core_2025} addresses (parameterized) complexities related to the core of ordered graphs and hypergraphs problems, where a core is defined as an ordered graph that is not homomorphic to a proper ordered subgraph. However, the complexity of determining whether an ordered graph is a core remains an unresolved question.

Ordered matchings, defined as ordered graphs where each vertex has exactly one incident edge, are known to play a crucial role in the study of ordered graphs and their homomorphisms (~\cite{balko2022offdiagonal}, ~\cite{Balko_2020}, ~\cite{conlon2016ordered}, ~\cite{nescer2023duality}). The complexities and parameterized complexities of the associated problems were determined in ~\cite{certik_matching_2025}.

In this article, we therefore try to extend the research by results on the complexity and parameterized complexity of fundamental problems related to homomorphisms of ordered graphs.

In Section ~\ref{Sec:Structures}, we address the construction of a reduction from unordered structures to ordered homomorphisms. We show that this construction is, in fact, feasible by using ordered bipartite graphs. 

In the beginning of Section ~\ref{Sec:complexity}, we show that for fixed ordered graph $H$, the $H$-coloring problem for ordered graph, analogous to the \NP-complete problem for unordered graphs (see ~\cite{HellNesetrilHColoring1990}), can be solved in polynomial time. We then define the problem \homo, as the problem of finding ordered homomorphism between given ordered graphs $G, H$ (see the problem definition ~\ref{Prb:ColouringOfG}), and we use the result from the previous section to prove the \NP-completeness of \homo.

Although, on the one hand, the existence of optimal coloring is polynomial (see also ~\cite{nescer2023duality}), we show the \NP-completeness of the general question in Section ~\ref{Sec:complexity}. We also show that the general question parameterized by $H$ is in \XP, which shows that with fixed $H$, finding the ordered homomorphism is in $P$.

Section ~\ref{Sec:ParamComplxt} then follows with the main result on the parameterized complexity of \homo$(|V(H)|)$ (\homo parameterized by $|V(H)|$) being \wone-hard.

In Section ~\ref{sec:polynomialcomp}, we then present two examples of ordered graphs classes, for which \homo problem can be solved in polynomial time.

The paper ends with several open problems and outlines of related research.

\section{Unordered Structures to Ordered Graphs Reduction}
\label{Sec:Structures}

    The reduction of unordered homomorphisms to ordered ones is of significant importance and represents one of the major results in our study. We try to undertake this reduction with full generality, ensuring its applicability and utility in the unfolding narrative of our work.

	A signature $\sigma$ consists of a finite set of relation symbols with specified arities. The arity of the relation $R$ is denoted by $\arity(R)$.
	A \emph{structure} $\calH$ with signature $\sigma(\calH)$ consists of a universe $V(\calH)$ together with a set of relations $\boldR(\calH)=\{R(\calH) \mid R\in \sigma(\calH)\}$ over the universe $V(\calH)$. We write $\calH=(V(\calH), \boldR(\calH))$.
	By $||\calH||$ we denote $|\sigma(\calH)| + |V(\calH)| + \sum_{R\in \sigma(\calH)} |R(\calH)|\cdot\arity(R)$, which is the size of a ``reasonable'' encoding of $\calH$.

    Given two structures $\calG$ and $\calH$ with the same signature $\sigma$, a function $f :  V(\calG) \to V(\calH)$ \emph{respects} $R\in \sigma$ if, for each $\boldx\in R(\calG)$, $f(\boldx)\in R(\calH)$, where $f$ is evaluated elementwise.  A \emph{homomorphism} from $\calG$ to $\calH$ is a function $h: V(\calG) \to V(\calH)$ that respects every $R\in \sigma$.

    If there exists a homomorphism from structure $\calG$ to structure $\calH$, we denote it by $\calG\to \calH$.

    \begin{theorem}\label{thm:structures}
        Given two unordered structures $\calG,\calH$ of the same signature $\sigma$, in time polynomial in $||\calG||+ ||\calH||$, we can construct a pair of ordered bipartite graphs $G,H$, such that $\calG \to \calH$ if and only if $G \to H$.
    \end{theorem}
    \begin{proof}    
        Without loss of generality, we can assume that every $v \in V(\calG)$ appears in some $\boldx \in \bigcup_{R \in \sigma} R(\calG)$, as otherwise we can safely remove $v$ from $\calG$ obtaining an equivalent instance of our problem.
        Fix some arbitrary ordering on the elements of $\sigma$.
        Furthermore, we fix an arbitrary ordering on the elements of $V(\calG)$, and an arbitrary ordering on the elements of $V(\calH)$.
        Finally, for all $R \in \sigma$, we fix an arbitrary ordering of the elements of $R(\calG)$ and of $R(\calH)$.

        \medskip
        We will perform the construction in two steps: first, we will build a pair of ordered bipartite graphs $G,H$, such that every vertex $v$ of $G$ is equipped with a list $L(v) \subseteq V(H)$, and $G$ has a homomorphism $h$ to $H$ respecting lists $L$ if and only if $\calG \to \calH$. By ordered homomorphism $h:G\to H$ \emph{respecting the lists $L$}, we will mean that for each $v\in V(G)$, it holds that $h(v)\in L(v) \subseteq V(H)$.

        \paragraph*{Definition of $H$.}
        
        The set $V(H)$ is partitioned into sets $A_H$ and $B_H$ which form the bipartition of $H$.
        The set $A_H$ contains a vertex $(v,u)$ for every $v \in V(\calG)$ and $u \in V(\calH)$.
        The set $B_H$ contains a vertex $(R,\boldx,\boldy)$ for every $R \in \sigma$, $\boldx \in R(\calG)$ and $\boldy \in R(\calH)$.
        Vertices $(v,u)$ and $(R,\boldx,\boldy)$ are adjacent in $H$ if and only if there exists $i \in [\arity(R)]$ such that $v = \boldx_i$ and $u = \boldy_i$.
        
        The ordering of $V(H)$ is defined as follows. All vertices of $A_H$ precede all vertices of $B_H$. The vertices within $A_H$ are ordered lexicographically according to the orderings of $V(\calG)$ and $V(\calH)$, i.e., $(v,u) < (v',u')$ if and only if $v<v'$ or $v=v'$ and $u<u'$.
        Similarly, the vertices in $B_H$ are ordered lexicographically according to the orderings of the sets $\sigma$, $R(\calG)$, and $R(\calH)$.
        
        \paragraph*{Definition of $G$.}

        The graph $G$ is defined as the \emph{incidence graph} of $\calG$.
        More specifically, its vertex set consists of two independent sets $A_G$ and $B_G$,
        where $A_G$ is $V(\calG)$ and $B_G$ is $\bigcup_{R \in \sigma} \bigcup_{\boldx \in R(\calG)} \{(R,\boldx)\}$.
                
        In the ordering of vertices, first we have all vertices from $A_G$ ordered according to their ordering in $\calG$, and then all the vertices from $B_G$, ordered lexicographically according to the orderings of $\sigma$ and $R(\calG)$, respectively.

        Finally, we set the list of every $v \in A_G$ to $\bigcup_{u \in V(\calH)} \{(v,u)\}$,
        and the list of every $(R,\boldx) \in B_G$ to $\bigcup_{\boldy \in R(\calH)} \{  (R,\boldx,\boldy)\}$. This completes the definition of $G$ and $H$.
        
        Note that $|V(G)|=|V(\calG)| + \sum_{R \in \sigma} |R(\calG)|$,
        and $|V(H)|=|V(\calG)| \cdot |V(\calH)| + \sum_{R \in \sigma} |R(\calG)| \cdot |R(\calH)|$, and these graphs can be constructed in time polynomial in $||\calG||+ ||\calH||$.

    \paragraph*{Equivalence of instances.}
        First, suppose that there is a homomorphism $f : \calG \to \calH$.
        We define a mapping $h: V(G) \to V(H)$ as follows.
        For each $v \in A_G$, we set $h(v)=(v,f(v))$.
        For each $(R,\boldx) \in B_G$, we set $h((R,\boldx))=(R,\boldx,f(\boldx))$, where $f(\boldx)$ is evaluated element-wise.

        It is clear that $h$ respects the lists $L$ and the orderings of $V(G)$ and $V(H)$. Respecting the lists can be seen from the definition of $h$, and $h$ respecting the ordering follows from the lexicographic ordering in the definition of $G$ and $H$.
        Furthermore, $h$ is a homomorphism (it respects (binary) relations), since $f$ is.

        On the other hand, suppose that $h$ is an order-preserving homomorphism from $G$ to $H$ that respects the lists $L$.
        We define the function $f : V(\calG) \to V(\calH)$ by mapping each $v \in V(\calG)$ to $u \in V(\calH)$ such that $h(v)=(v,u)$.
        We claim that $f$ is a homomorphism from $\calG$ to $\calH$.
        Taking some $R \in \sigma$ and $\boldx \in R(\calG)$, we aim to show that $f(\boldx) \in R(\calH)$.
        Note that $h((R,\boldx))$ must be a vertex $(R,\boldx,\boldy)$ of $H$ which is a common neighbor of $\bigcup_{v \in \boldx} h(v)$.
        By construction of $H$, such a vertex exists if and only if $\boldy \in R(\calH)$. Consequently, we have $f(\boldx) = \boldy \in R(\calH)$.

        This shows that the constructed instances are indeed equivalent.
        
    \paragraph*{Non-list variant.}
       We modify $G$ as follows. Let $p=|V(\calG)|$ and $q=|\bigcup_{R \in \sigma} R(\calG)|$. We modify $G$ into $G'$ by adding a separate component, which is a path $P$ on $(p+1)+1+(q+1)$ vertices denoted consecutively by $x_0,\ldots,x_p,y,z_0,\ldots,z_q$.
       In the ordering of $V(G')$ the vertex $x_0$ is the first one, then we insert $x_i$ immediately after the $i$-th vertex from $A_G$. The vertex $x_p$ is succeeded by $y$ and $z_0$, and then we insert $z_i$ after the $i$-th vertex from $B_G$.

       In an analogous way, we modify $H$ into $H'$: we introduce a new component, which is a path $P'$ with consecutive vertices $x'_0,\ldots,x'_p,y',z'_0,\ldots,z'_q$.
       The vertex $x'_0$ is the first vertex of $H'$.
       The vertex $x'_i$ for $i \in [p]$ is inserted immediately after the last vertex from $L(v)$, where $v$ is the $i$-th vertex in $A_G$. 
       The vertex $x'_p$ is succeeded by $y'$ and then by $z'_0$.
       Then we insert each $z'_i$ for $i \in [q]$ after the last vertex of $L((R,\boldx))$, where $(R,\boldx)$ is the $i$-th vertex from $B_G$.

       We argue that $G' \to H'$ if and only if $G$ admits a homomorphism to $H$ that preserves lists $L$.
        If such a homomorphism from $G$ to $H$ exists, we can easily extend it to the homomorphism from $G'$ to $H'$ by mapping every $x_i$ to $x'_i$, $y$ to $y'$ and every $z_i$ to $z'_i$.

        Now suppose that $h$ is a homomorphism from $G'$ to $H'$. We claim that $h$ restricted to $V(G)$ is a homomorphism from $G$ to $H$ that preserves the lists $L$.

        Let us now define a \emph{forward path} as a path whose vertices are ordered in a natural way. Let $G_1,G_2$ be ordered graphs, and let $g$ be an order-preserving homomorphism from $G_1$ to $G_2$. Let $P_1$ be a forward path in $G_1$. Then we observe that $g$ is injective on $P_1$ and the image of $P_1$ contains a spanning forward path.

        Observe that the longest forward path in $H$ has two vertices and $P$ has at least five vertices. Consequently, by the previous observation, $h$ maps each vertex from $P$ to its primed counterpart.
        Notice that if we now ensure that each vertex of $G$ is mapped to a vertex of $H$, we know that $h$ must respect lists $L$.

        For contradiction, suppose that some vertex of $G$ was mapped to $P'$.
        Notice that since every $v \in V(\calG)$ appears in some relation $R(\calG)$, 
        $G$ has no isolated vertices.
        As the whole connected component of $G'$ must be mapped to the same connected component of $H'$, we conclude that there are some vertices $v \in A_G$ and $(R,\boldx) \in B_G$ that are adjacent in $G$ that are mapped to $P'$.
        Say $v$ is the $i$-th vertex of $A_G$ and $(R,\boldx)$ is the $j$-th vertex of $B_G$.
        Since $P$ is mapped to $P'$, we conclude that $v$ is mapped either to $x'_{i-1}$ or to $x'_i$, and $(R,\boldx)$ is mapped either to $z'_{j-1}$ or to $z'_j$.
        However, all these vertices are pairwise nonadjacent in $P'$, contradicting the fact that $h$ is a homomorphism.
        This completes the proof.
    \end{proof}

You can find the application and an example of this reduction in Corollary ~\ref{cor:HomNPC}.

\section{Complexities of Finding Homomorphisms of Ordered Graphs}
\label{Sec:complexity}

In this section, we shall focus on what we consider as natural problems arising when studying complexities of ordered homomorphisms.

\subsection{Polynomial-time $H$-coloring}

Analogously to the $H$-coloring problem of unordered graphs, let us start with the following problem \homoH, assuming a fixed ordered graph $H$.

\begin{prb}
\label{Prb:HColouringOfG}
\end{prb}
\problemStatement{\homoH}
  {Ordered graph $G$.}
  {Does there exist an ordered homomorphism from $G$ to the fixed ordered graph $H$?}

We consistently denote $n = |V(G)|$ and $h=|V(H)|$. The following complexity then follows easily.

\begin{prop}
\label{prop:XP}
    The \homoH problem can be solved in time $\Oh(n^{h-1})$.
    Consequently, for every fixed ordered graph $H$, \homoH is in \PP.
\end{prop}

\begin{proof}
    The statement follows since we can simply count all the $\binom{n+h-1}{h-1}$ possible mappings that preserve the ordering (this can be seen as the number of possible solutions to $x_1+\ldots+x_h=n$ with $x_i\in\{0,\ldots,n\}$), which for fixed $h$ is $\Oh(n^{h-1})$.
\end{proof}

This is, of course, in sharp contrast to the $H$-coloring problem of unordered graphs, which is \NP-complete for any non-bipartite graph $H$ (see ~\cite{HellNesetrilHColoring1986}).

Similarly to unordered graphs, we can also consider the problem of minimum coloring of $G$ and define \emph{the (ordered) chromatic number} $\chi^<(G)$ to be the minimum $k$ such that $V(G)$ can be partitioned into $k$ disjoint independent intervals. Notice that for ordered graphs this is the size of the smallest homomorphic image and, alternatively, the minimum $k$ such that $G\to K_k$, $K_k$ being a complete graph with fixed linear ordering. We shall also call $\chi^<(G)$ a \emph{coloring} of $G$.


We have shown the determination of $\chi^<(G)$ using a simple greedy algorithm in ~\cite{nescer2023duality}, and the topic of $\chi^<$-boundedness of ordered graphs is covered in ~\cite{Axenovich2016ChromaticNO}, ~\cite{nescer2023duality}. Again, this is in contrast with the similar notion for unordered graphs, which is \NP-complete for $K_i, i>2$. 

\subsection{\NP-Completeness of General Problem}

Let us now generalize the \homoH problem and start with a definition of the computational problem we refer to as \homo, whose input is a pair of ordered graphs $G$ and $H$ and we ask if $G$ admits an ordered homomorphism to $H$.

\begin{prb}
\label{Prb:ColouringOfG}
\end{prb}
\problemStatement{\homo}
  {Ordered graphs $G$ and $H$.}
  {Does there exist ordered homomorphism from $G$ to $H$?}

The following result then follows from Theorem ~\ref{thm:structures}.

    \begin{corollary}
    \label{cor:HomNPC}
    The problem \homo is \NP-complete. Furthermore, there is no algorithm that solves every instance $G,H$ of \homo in time subexponential in $|V(G)| + |V(H)|$, unless the ETH fails.
\end{corollary}
        \begin{proof}
            Let $\calG$ be an instance of 3-\textsc{Coloring} of maximum degree 4. We know from ~\cite{doi:10.1137/1.9781611974331.ch112}, that this problem is \NP-hard and has ETH lower bound.
            
            We can see this problem as the homomorphism problem of structures with signature $\{E\}$, where the target structure $\calH$ is of constant size.
            We invoke \cref{thm:structures} to obtain in polynomial time an equivalent instance $(G,H)$ of \homo, which proves that \homo is \NP-hard.

            Note that the number of vertices of $G$ is $|V(\calG)|+|E(\calG)|=\Oh(|V(\calG)|)$,
            and the number of vertices of $H$ is $3|V(\calG)| + 6|E(\calG)| = \Oh(|V(\calG)|)$.
            Consequently, the lower bound of the ETH holds.            
            \end{proof}

\section{Parameterized Complexity}
\label{Sec:ParamComplxt}

We shall continue by investigating a parameterized complexity of \homo. 

Similarly to unordered graphs, one of the interesting parameters for exploring the parameterized complexity of \homo is $h=V(H)$. We therefore denote \homo parameterized by $h$ by \parhomo.

To justify exploring the parameterized complexity of \parhomo, we provide its following parameterized complexity upper bound, following from the proof of Proposition ~\ref{prop:XP}.

\begin{prop}
\label{cor:XP}
    The \parhomo problem is in \wpe.
\end{prop}

\begin{proof}
    According to the definition of class \wpe, a parameterized problem, parameterized by $k$, is in \wpe if it can be solved by a nondeterministic algorithm that runs in time $f(k)\cdot n^{\Oh(1)}$ and uses at most $g(k)\cdot \log n$ nondeterministic bits.
    
    We will show that for \parhomo, where $k=|V(H)|$, such an algorithm exists.
    
    Indeed, from the proof of Proposition ~\ref{prop:XP}, a homomorphism $f:V(G)\to V(H)$ can be represented by the numbers $a_1, a_2, \ldots, a_k$, where $a_i$ is the number of vertices of $G$ mapped to the $i$-th vertex of $H$. Each $a_i\in [0,n]$ requires $\log n$ bits, so the entire mapping can be nondeterministically guessed using $k\log n$ bits. We can see that this mapping preserves the ordering.

    The preservation of edges can then be verified in polynomial time by checking that each edge in $G$ is mapped to an edge in $H$ (which takes time $\Oh(|E(G)|)$, which is polynomial in $|V(G)|=n$).
\end{proof}

Let us now prove the main parameterized complexity result for \homo$(|V(H)|)$.

\begin{theorem}\label{thm:wone}
    The \homo$(|V(H)|)$ problem is \wone-hard.
    Furthermore, it cannot be solved in time $n^{o(h)}$, unless the ETH fails.
\end{theorem}

\begin{proof}
Similarly as in the proof of Theorem~\ref{thm:structures}, we will prove this theorem in two steps.
First, let us show the hardness in the \emph{list} setting: Each vertex $v$ of $G$ is equipped with a list $L(v) \subseteq V(H)$, and we additionally require that the homomorphism $f$ we are looking for satisfies $f(v) \in L(v)$ for every $v$.

We reduce from \textsc{Multicolored Independent Set}. Let $F$ be the instance graph whose vertex set is partitioned into $k$ subsets $V_1,V_2,\ldots,V_k$. We ask if $F$ has an independent set of size $k$, intersecting every set $V_i$. By copying some vertices if necessary, without loss of generality, we may assume that for each $i \in [k]$ we have $|V_i|=\ell$. 
The problem \textsc{Multicolored Independent Set} cannot be solved in time $(k\ell)^{o(k)}$, unless the ETH fails, and the problem is \wone-hard, parameterized by $k$ (see, e.g.,~\cite{dvorak2023parameterized-d6f,bonnet2020parameterized-4f9,cygan2015parameterized-4b3}).

\paragraph*{Definition of $H$.}
The graph $H$ has $5k$ vertices $\bigcup_{i\in[k]} \{a_i,b_i,x_i,y_i,z_i\}$, ordered as follows:
\[
a_1,b_1,a_2,b_2,\ldots,a_k,b_k,x_1,y_1,z_1,x_2,y_2,z_2,\ldots,x_k,y_k,z_k.
\]
For each $i \in [k]$, we add edges $a_ix_i, a_iy_i, b_iy_i, b_iz_i$. Furthermore, we add all edges in the set $\bigcup_{i\in[k]} \{x_i,y_i,z_i\}$, except for the edges $y_iy_j$. The vertices $\bigcup_{i \in [k]} \{y_i\}$ form an independent set in $H$.
This completes the definition of $H$.

\paragraph*{Definition of $G$.}
Fix $i \in [k]$ and an arbitrary total order on $V_i = \{v^i_1,\ldots,v^i_{\ell}\}$.
We introduce to $G$ a set $P^i$ with vertices $p^i_0,p^i_1,\ldots,p^i_\ell$ (with such an ordering).
We set lists $L(p^i_0)=\{a_i\}$, $L(p^i_\ell)=\{b_i\}$, and $L(p^i_j) = \{a_i,b_i\}$ for all $j \in [\ell-1]$.

Note that in any order-preserving mapping from $P^i$ to $H$ there is exactly one $j \in [\ell]$ such that $p^i_{j-1}$ is mapped to $a_i$ and $p^i_j$ is mapped to $b_i$. We will interpret choosing such a mapping as choosing $v^i_j$ to the solution (that is, the independent set in $F$).

Next, we introduce a set $Q^i$ of $\ell+2$ vertices $q^i_0,q^i_1,\ldots,q^i_\ell,q^i_{\ell+1}$ (with such an ordering).
We set lists $L(q^i_0)=\{x_i\}, L(q^i_{\ell+1})=\{z_i\}$, and $L(q^i_j)=\{x_i,y_i,z_i\}$ for $j \in [\ell]$.
For $j \in [\ell]$, the vertex $q^i_j$ is adjacent to $p^i_{j-1}$ and $p^i_j$.
We observe that if $p^i_{j-1}$ is mapped to $a_i$ and $p^i_{j}$ is mapped to $b_i$ (i.e., $v^i_j$ is selected to the independent set),
then $q^i_j$ must be mapped to $y_i$.
All vertices $q^i_{j'}$ for $j' < j$ are mapped to $x_i$ or $y_i$ (and it is possible to map them all to $x_i$).
Similarly, all vertices $q^i_{j'}$ for $j' > j$ are mapped to $y_i$ or $z_i$ (and it is possible to map them all to $z_i$).

Finally, for $j,j' \in [\ell]$ and distinct $i,i' \in [k]$, we add an edge $q^i_jq^{i'}_{j'}$ if an only if $v^i_j$ is adjacent to $v^{i'}_{j'}$.\footnote{Note that we can assume that in our instance of \textsc{Multicolored Independent Set} each set $V_i$ is independent. Then the subgraph of $G$ induced by $\bigcup_{i \in [k]} \bigcup_{j \in [\ell]} \{q^i_j\}$ is isomorphic to $F$.}

We order the vertices of $G$ as follows:
\[
P^1, P^2, \ldots, P^k, Q^1, Q^2, \ldots, Q^k,
\]
where the ordering within each set is as specified above. $G$ has $k(\ell+1) + k(\ell+2)=2k\ell + 3k$ vertices and can clearly be constructed in polynomial time. This completes the definition of $G$.

\paragraph*{Equivalence of instances.} Suppose that $F$ has a yes instance of \textsc{Multicolored Independent Set}, that is,
for each $i \in [k]$ there is $j_i$, such that $I = \bigcup_{i \in [k]} \{ v^i_{j_i}\}$ is an independent set in $F$.

We now define $f : V(G) \to V(H)$. Fix $i \in [k]$. We map $p^i_j$ to $a_i$ if $j < j_i$ and to $b_i$ otherwise.
We map $q^i_j$ to $x_i$ if $j < j_i$, to $y_i$ if $j = j_i$, and to $z_i$ if $j > j_i$.

Clearly, $f$ respects lists and the ordering of vertices.
Let us discuss that it preserves edges. As discussed above, all edges between sets $P^i$ and $Q^i$ are mapped to edges of $H$.
So let us consider an edge $q^i_jq^{i'}_{j'}$. For contradiction, suppose that its image is a non-edge of $H$, which means that $q^i_j$ is mapped to $y_i$ and $q^{i'}_{j'}$ is mapped to $y_{i'}$. This means that the vertices $v^i_j$ and $v^{i'}_{j'}$ are in $I$. However, they are adjacent in $F$ as $q^i_jq^{i'}_{j'}$ is an edge in $G$, a contradiction.

\medskip
Now consider a homomorphism $f: V(G) \to V(H)$, which respects lists and the ordering of vertices.
As observed before, for each $i \in [k]$ there is exactly one $j_i \in [\ell]$ such that $f(p^i_{j_i-1})=a_i$ and $f(p^i_{j_i})=b_i$.
Let $I = \bigcup_{i \in [k]} \{v^i_{j_i}\}$, we claim that $I$ is independent in $F$.
For contradiction, suppose that $v^i_{j_i}$ is adjacent to~$v^{i'}_{j_{i'}}$. This means that $q^i_{j_i}$ is adjacent to $q^{i'}_{j_{i'}}$ by the definition of $G$.
However, $f(q^i_{j_i})=y_i$ (since no other vertex of $\{x_i,y_i,z_i\}$ is adjacent to both $a_i$ and $b_i$) and analogously $f(q^{i'}_{j_{i'}})=y_{i'}$.
As $y_i$ is not adjacent to $y_{i'}$ in~$H$, we obtain a contradiction.

\medskip

This completes the proof of the list version of the theorem. The astute reader might notice that vertices $q^i_0$ and $q^i_{\ell+1}$ do not play any role in the above reasoning. However, they will be useful in the next part, when we show hardness in the non-list setting.

\paragraph*{Hardness in the non-list setting.}
Let $G$ and $H$ be the graphs obtained so far. We will modify them to obtain $G'$ and $H'$, respectively, such that $G'$ admits a homomorphism to $H'$ if and only if $G$ admits a homomorphism to $H$ which respects lists $L$.

Observe that the largest clique in $H$ has $2k+1$ vertices: It consists of all vertices in $\bigcup_{i \in [k]} \{x_i,z_i\}$ and one vertex $y_i$.
We introduce $2k+2$ new vertices $c_1,c_2,\ldots,c_{2k+2}$ that form a clique.
For each $i \in [k]$, we make $c_i$ adjacent to $a_i$, and $c_{i+1}$ adjacent to $b_i$.
Furthermore, for each $i \in [k]$, we make $c_{k+i}$ adjacent to $x_i$, and $c_{k+i+1}$ adjacent to $z_i$.
The ordering of vertices of $H'$ is defined as follows:
\[
c_1,a_1,b_1,\ldots,c_k,a_k,b_k,c_{k+1},x_1,y_1,z_1,c_{k+2},\ldots,c_{2k},x_k,y_k,z_k,c_{2k+1}, c_{2k+2}.
\]

We modify $G$ similarly: we add a clique on $2k+2$ vertices $r_1,r_2,\ldots,r_{2k+2}$.
For each $i \in [k]$, we make $r_i$ adjacent to $p^i_0$, and $r_{i+1}$ adjacent to $p^i_\ell$.
Furthermore, for $i \in [k]$, we make $r_{k+i}$ adjacent to $q^i_0$, and $r_{k+i+1}$ adjacent to $q^i_{l+1}$, where the ordering within the sets is as before.

We need to show that any homomorphism $f : V(G') \to V(H')$, restricted to $V(G)$, respects lists~$L$.
Note that the only $(2k+2)$-clique in $H'$ consists of vertices $\bigcup_{i \in [2k+2]} \{c_i\}$.
Consequently, the $(2k+2)$-clique $\bigcup_{i \in [2k+2]} \{r_i\}$ of $G'$ must be mapped by $f$ to this clique.
Furthermore, the only way to achieve this while preserving the ordering of the vertices is to map every $r_i$ to $c_i$ (for $i \in [2k+2]$).

Consider a vertex $p^i_0$ for $i \in [k]$.
Since $f$ is order-preserving, we have $f(p^i_0) \in \{c_i,a_i,b_i,c_{i+1}\}$.
However, $p^i_0$ is adjacent to $r_i$, and thus $f(p^i_0) \in \{a_i, c_{i+1}\}$.
If $f(p^i_0)=c_{i+1}$, then we also have $f(p^i_j)=c_{i+1}$ for all $j \in [\ell]$.
This is a contradiction, as $p^i_\ell r_{i+1} \in E(G')$ and $f(r_{i+1})=c_{i+1}$. Consequently, $f(p^i_0)=a_i$.
Analogously, we can argue that $f(p^i_\ell)=b_i$.
Thus, again using that $f$ is order-preserving, we obtain that $f(p^i_j) \in \{a_i,b_i\}$ for all $j \in [\ell-1]$.

In exactly the same way, we argue that for all $i \in [k]$ we have $f(q^i_0)=x_i$ and $f(q^i_{\ell+1}) = z_i$,
and consequently for all $j \in [\ell]$ we have $f(q^i_j) \in \{x_i,y_i,z_i\}$.

As $V(H') = 5k + 2k+2 = 7k+2$ and $V(G') = 2k\ell + 3k + 2k+2 = 2k \ell + 5k +2$, the claimed hardness follows. This completes the proof.
\end{proof}

\section{Polynomial-time cases of \homo}

\label{sec:polynomialcomp}

In Section ~\ref{Sec:complexity} we have determined that the complexity of the general \homo problem is \NP-complete. In this section, we shall focus on parameters and classes of ordered graphs, for which \homo and its variations are solvable in polynomial time.


\subsection{$k$-shifted Cliques}

We start with a definition of a class of ordered graphs $\mathcal{H}$, for which the \homo problem, where $H\in\mathcal{H}$, is polynomial.

We call an ordered graph $H$ a \emph{$k$-shifted clique} if:
\begin{itemize}
    \item $V(H)$ can be partitioned into segments $V_1,\ldots,V_k; V_i < V_{i+1}, i\in[k-1]$, each of which induces a clique in $H$,
    \item For $ i> 1$ and each $u,u' \in V_i$ such that $u < u'$, it holds that $N(u') \cap \bigcup_{j <i} V_j \subseteq N(u) \cap \bigcup_{j <i} V_j$.    
\end{itemize}

In fact, we will derive the result for a more general problem, which we call \homos.

\begin{prb}
\label{Prb:homos}
\end{prb}
\problemStatement{\homos}
  {Ordered graph $H$ that is a $k$-shifted clique and an ordered graph $G$ equipped with functions $\mathsf{low}, \mathsf{up} : V(G) \to [h]; h=|V(H)|$.}
  {Does there exist a homomorphism $f : G \to H$ such that for every $v \in V(G)$ we have $\mathsf{low}(v) \leq f(v) \leq \mathsf{up}(v)$.}

We aim to show the following result:
\begin{theorem}\label{thm:thickpaths}
    If $H$ is a $k$-shifted clique, then the \homos problem can be solved in time $n^{\Oh(k)} \cdot h^{\Oh(1)}$, where $n = |V(G)|$ and $h = |V(H)|$.
\end{theorem}
Clearly, \homos is again a list homomorphism problem and setting $\mathsf{low} \equiv 1$ and $\mathsf{up} \equiv h$ gives us a polynomial-time algorithm for \homo.

We say that a solution $f$ to \homos (that is, an ordered homomorphism respecting functions $\mathsf{low}$ and $\mathsf{up}$) is \emph{minimum} if, for any solution $f'$ and any $v \in V(G)$, it holds that $f(v) \leq f'(v)$.

The base case is for $k=1$, that is, if $H$ is a clique. We will now present two lemmas.

\begin{lemma}\label{lem:clique-with-lists}
    One can solve an instance of \homos where $G$ has $n$ vertices and $H = K_h$ in time $n^{\Oh(1)}$.
    Furthermore, if there is a solution, one can even find a minimum solution $f$.
\end{lemma}
\begin{proof}
    To prove the statement, we can exhaustively perform the following two update operations.
    For an edge $vv' \in E(G)$ such that $v < v'$, set $\mathsf{low}(v') = \max (\mathsf{low}(v'), \mathsf{low}(v)+1)$.
    For pair $v$ and $v'\in V(G)$, such that $v < v'$, set $\mathsf{low}(v') = \max (\mathsf{low}(v'), \mathsf{low}(v))$.
    We terminate the procedure when no further changes are made.
    Clearly, the procedure takes time polynomial in $n$ and $h$ (in particular, $\Oh (n^2)$, since we iterate over all pairs of vertices).

    If there exists a vertex $v \in V(G)$ such that $\mathsf{low}(v) > \mathsf{up}(v)$, we terminate the algorithm and answer that no solution exists.
    Otherwise, we set $f = \mathsf{low}$.


    Let us argue that $f$ has all the desired properties.
    First, notice that by the first update operation, $f$ is a proper coloring (since $H$ is a complete graph) and by the second update operation, it is monotone.
    Furthermore, the value of $f$ is always at least the initial value of $\mathsf{low}$, as during the procedure we only increased the values.
    Finally, $f(v) \leq \mathsf{up}(v)$, as otherwise, we terminate the algorithm. Therefore, $f$ is a solution of \homos.

    The minimality of $f$ can be easily shown by induction on the ordering of $V(G)$. This is clearly true for the first vertex $v_1$ of $V(G)$, since $\mathsf{low}(v_1)$ is not adjusted by two update operations. By the induction hypothesis, $\mathsf{low}(v_i), i=1,\ldots,k$ is minimal and, by two update operations, we have $\mathsf{low}(v_i)\le \mathsf{low}(v_k),i=1,\ldots,k-1$. By two update operations also $\mathsf{low}(v_i)\le \mathsf{low}(v_{k+1}),i=1,\ldots,k$. We have $\mathsf{low}(v_k)= \mathsf{low}(v_{k+1})$ if $v_{k+1}$ is not adjacent to any of the vertices mapped to $\mathsf{low}(v_{k})$ and if the initial $\mathsf{low}(v_{k+1})$ is lower than or equal to the final $\mathsf{low}(v_k)$. On the other hand, it is $\mathsf{low}(v_k)< \mathsf{low}(v_{k+1})$, if $v_{k+1}$ is adjacent to some of the vertices that mapped to $\mathsf{low}(v_{k})$ or if the initial $\mathsf{low}(v_{k+1})$ is higher than the final $\mathsf{low}(v_k)$. Since the final $\mathsf{low}(v_{k+1})$ is minimal with respect to its initial value and the preceding (adjacent and nonadjacent) vertices $v_i,i=1,\ldots,k$, its minimality follows.
\end{proof}

A similar result to the Lemma~\ref{lem:clique-with-lists} and an algorithm to find a minimum solution $f$ can also be found in~\cite{nescer2023duality}.

\begin{lemma}\label{lem:replacement}
    Consider an instance $(G,H,\mathsf{low},\mathsf{up})$ of \homos, where $H$ is a $k$-shifted clique for $k \geq 2$.
    Let $R= \{v_p,\ldots,v_h\}$ be the last of the segments that induce a clique (that is, $R = V_k$). Let $L = V(H) \setminus R$.

    Suppose that there exists a solution $\phi : G \to H$ and let $X = \phi^{-1}(R)$.
    Let $f$ be the minimum solution to \homos on the instance $(G[X],H[R],\mathsf{low}',\mathsf{up}')$, where the functions $\mathsf{low}',\mathsf{up}' : X \to \{p,\ldots,h\}$ are defined as follows,
    
    \begin{align*}
    \mathsf{low}'(u) &= \max (\mathsf{low}(u), p) \\
    \mathsf{up}'(u) &= \mathsf{up}(u).
    \end{align*}
    
    Then the function $\phi' : V(G) \to V(H)$ defined as
    \[
        \phi'(u) = \begin{cases}
            f(u) & \text{ if } u \in X,\\
            \phi(u) & \text{ if } u \notin X.
            \end{cases}
    \]
    is a solution to the original instance $(G,H,\mathsf{low},\mathsf{up})$ of \homos.    
\end{lemma}
\begin{proof}
    Before we proceed to the proof, let us point out that $f$ exists, as $\phi|_X$ is a solution to $(G[X],H[R],\mathsf{low}',\mathsf{up}')$. This can be seen since $\mathsf{low}(u),u\in X$ cannot be lower than $p$, by the definition of $X=\phi^{-1}(R)$.

    Notice that $\phi'$ is monotone, as each of $\phi,f$ is monotone, and the codomain of $\phi$ is $\{1,\ldots,p-1\}$ while the codomain of $f$ is $\{p,\ldots,h\}$.
    Furthermore, $\phi'$ respects $\mathsf{low}$ and $\mathsf{up}$, as $\phi$ respects them,
    $\mathsf{up}' = \mathsf{up}|_X$, and $\mathsf{low}'$ is only more restrictive than $\mathsf{low}$.

    So we need to show that $\phi'$ is a homomorphism. Consider an edge $uu' \in E(G)$ such that $u < u'$. If $u,u' \in X$ or $u,u' \notin X$, then the claim follows, as both $\phi$ and $f$ are homomorphisms. Suppose that $u \notin X$ and $u' \in X$.
    Let $v_i = \phi'(u) = \phi(u)$ and $v_j = f(u')$.
    Observe that by the minimality of $f$, we have $v_j \leq \phi(u')$.
    On the other hand, we have $v_i \in N(\phi(u')) \cap L$.
    By the properties of $k$-shifted cliques, we obtain $v_i \in N(v_j) \cap L$.
    Thus, indeed, $\phi'$ is a solution to the instance $(G,H,\mathsf{low},\mathsf{up})$ of \homos.
\end{proof}

Now we are ready to prove Theorem~\ref{thm:thickpaths}.
\begin{proof}[of Theorem~\ref{thm:thickpaths}.]
    We proceed by induction on $k$.
    If $k = 1$, we are done by Lemma~\ref{lem:clique-with-lists}.
    Suppose that $k \geq 2$ and that the result is true for all $k-1$.

    Let $R = V_k = \{v_p,\ldots,v_h\}$, that is, the last segment that induces a clique in $H$ and let $L = V(H) \setminus R$.    
    Note that $H' = H-R$ is a $(k-1)$-shifted clique.
    We first check if the instance admits a solution whose codomain is $L$ (therefore solution to the instance of $G$ and $H'$); we do it by calling the algorithm inductively (we need to modify the function $\mathsf{up}$ in the obvious way by setting $\mathsf{up}(u)=\min (\mathsf{up}(u), p-1), u\in G$).
    If so, we return it as a solution we seek and stop. We see that this gives us a complexity $n^{\Oh(k)}$, by the algorithm of Lemma~\ref{lem:clique-with-lists} being performed inductively $\Oh(k)$ times.

    Otherwise, we exhaustively guess the first vertex $v$ of $G$ mapped to a vertex in $R$, which gives $n$ possibilities. Let $X$ be the subset of $V(G)$ consisting of $v$ and all its successors.

     We call the algorithm from the proof of the Lemma~\ref{lem:clique-with-lists} for $G[X]$ and $H[R]$ with the function $\mathsf{low}$ modified as in Lemma~\ref{lem:replacement}.
     If it finds no solution, we terminate the current branch (of the loop over $n$).
     Suppose that it found the minimum solution $f$.

    We call the algorithm inductively for $G[V(G)\setminus X]$ and $H[L]$ and functions $\mathsf{low}',\mathsf{up}'$ obtained as follows. Consider $u \in V(G) \setminus X$ and $U = X \cap N(u)$.
    
    We set

    \begin{align*}
    &\mathsf{up}'(u) = \min (\mathsf{up}(u), p-1) \\
    &\mathsf{low}'(u) = \max (\mathsf{low}(u), \max_{u'\in U} \min N(f(u'))).
    \end{align*}

    This $\mathsf{low}'(u)$ adjustment is necessary because it ensures that $\mathsf{low}'(u)$ is connected to $f(u'), u' \in U$ (since $\mathsf{low}(u)$ could not be connected to all $f(u'), u' \in U$), while keeping $\mathsf{low}'(u)$ minimal (since it is the minimum vertex to which $f(u')$ is connected for all $u'\in U$ (by the properties of the $k$-shifted cliques)).

If it does not find a solution, we terminate the current branch.

If both calls find solutions, we return a mapping that is the union of two mappings found.
Note again that it is a solution by the properties of $k$-shifted cliques.

Since the second call takes time $h^{\Oh(1)}$, the overall complexity is $n^{\Oh(k)}h^{\Oh(1)}$.

On the other hand, as shown in Lemma~\ref{lem:replacement}, if there exists any solution, there is a solution of the kind that we seek by Lemma~\ref{lem:replacement}.
    \end{proof}

\subsection{Small generalized pathwidth}

In this subsection, we start with the definition of the parameter $c(H)$.

Let $H$ be an ordered graph with vertex set $(v_1,\ldots,v_h)$.
By $V_{\leq i}$ (resp. $V_{>i}$) we denote the set $\{v_1,\ldots,v_i\}$ (resp. $\{v_{i+1},\ldots,v_h\}$). Let $A_i$ consist of the vertices of $V_{\leq i}$ that have neighbors in $V_{>i}$. Similarly, let $B_i$ consist of the vertices of $V_{\leq i}$ that have non-neighbors in $V_{>i}$. Then, let 

$$c(H) = \max_i \min (|A_i|,|B_i|).$$

We point out that $\max_i  |A_i|$ is roughly equal to the pathwidth of $H$. We can approximately see this using the following construction. Each ordering of vertices of an unordered graph can define a path decomposition as follows. We start with a bag $B_1$ consisting of a single vertex $v_1$. Now suppose that we defined $B_1, \ldots, B_i$ and want to define $B_{i+1}$. We obtain it from $B_i$ by adding $v_{i+1}$ and removing all vertices whose all edges are already covered.

Optimum path decomposition can be obtained from some (optimal) ordering.
Since we deal with a fixed ordering, the parameter we obtain is not the
same as the pathwidth of the underlying unordered graph, but it behaves
similarly.


\begin{theorem}
    The \homo problem can be solved in time $n^{\Oh(c(H))} \cdot h^{\Oh(1)}$, where $n = |V(G)|$ and $h = |V(H)|$.
\end{theorem}

\begin{proof}
Let $k = c(H)$ and $V(H) = (v_1, \ldots, v_h)$.
For $i$, let $i'$ be the minimum value of $j \geq i$ such that $|A_j|\leq k$ (if such $j$ exists).
Similarly, let $i''$ be the minimum value of $j \geq i$ such that $|B_j|\leq k$ (again, if such $j$ exists).
 We define $\widetilde{A}_i = A_{i'} \cap V_{\leq i}$ and $\widetilde{B}_i = B_{i''} \cap V_{\leq i}$. If $i'$ (resp. $i''$) is not defined, we set  $\widetilde{A}_i = \emptyset$ (resp.  $\widetilde{B}_i = \emptyset$). Observe that for all $i$ we have $|\widetilde{A}_i| \leq k$ and $|\widetilde{B}_i| \leq k$ and for each $i$ we either have $A_i = \widetilde{A}_i$ or $B_i = \widetilde{B}_i$.

We will construct a dynamic programming procedure that computes partial homomorphisms from $G$ to $H$. The dynamic programming table $\textsf{Tab}$ is indexed by the tuple consisting of:
\begin{itemize}
    \item a vertex $u \in V(G)$,
    \item $i \in [h]$,
    \item a set $X \subseteq V(G)$, such that for each $x \in X$ we have $x \leq u$,
    \item a function $\xi$ mapping the vertices of $G$ smaller than or equal to $u$ to $\widetilde{A}_i \cup \widetilde{B}_i$.
\end{itemize}
Note that the number of choices for $X$ and $\xi$ is at most $n^{4k}$: the size of $|\widetilde{A}_i \cup \widetilde{B}_i|$ is at most $2k$ and the preimage of each element of $\widetilde{A}_i \cup \widetilde{B}_i$ is a segment that can be described by two vertices of $G$ (its first and last element), giving us at most $(n\cdot n)^{2k}$ options.

Thus, the total number of entries in $\textsf{Tab}$ is at most $n^{4k+1} \cdot h$.

We aim to set $\textsf{Tab}[u,i,X,\xi] = \textsf{true}$ if and only if there exists a homomorphism $f : G[X] \to H$, $G[X]$ being an induced ordered subgraph of $G$ on vertices $X$, such that:
\begin{itemize}
    \item $f(u) = i$, 
    \item $f^{-1}(\widetilde{A}_i \cup \widetilde{B}_i) = X$, and 
    \item $f|_X = \xi$.
\end{itemize} 
Clearly, $G \to H$ if and only if $\textsf{Tab}$ contains a true entry such that the first coordinate of the indexing tuple is the last vertex of $G$.

Now let us describe how we fill the table  $\textsf{Tab}$.

Let $u$ be the first vertex of $G$; then the entries for the first vertex of $G$ are straightforward to fill. For each $i \in [h]$, if $v_i \notin \widetilde{A}_i \cup \widetilde{B}_i$, we set $\textsf{Tab}[u,i,\emptyset,\emptyset] = \textsf{true}$.
Otherwise, we set $\textsf{Tab}[u,i,\{u\},\xi] = \textsf{true}$, where $\xi$ maps $u$ to $v_i$.

As mentioned, if $\textsf{Tab}$ contains a true entry such that the first coordinate of the indexing tuple is the last vertex of $G$, we get an ordered homomorphism $G\to H$.

For each $i \in [h]$ and each tuple $(u',i',X',\xi')$ for which $\textsf{Tab}[u',i',X',\xi']=\textsf{true}$ and $i' \leq i$, we proceed as follows.

\begin{itemize}
    \item Consider two subcases:
    \begin{itemize}
        \item If $\widetilde{A}_{i'} = A_{i'}$, for every neighbor $u''$ of $u$ that precedes it, 
        check whether $u'' \in X'$  and $\xi'(u'') \in N(v_i)$ (such $\xi'(u'')$ must be in $A_{i'}$).
        If not, proceed to the next tuple (since the tuple would not preserve edges).
        \item If $\widetilde{B}_{i'} = B_{i'}$, for every neighbor $u''$ of $u$ that precedes it, if $u'' \in X'$, check whether $\xi'(u'') \in N(v_i)$. If not, proceed to the next tuple (since the tuple would not preserve edges).
    \end{itemize}
    \item Consider two subcases:
    \begin{itemize}
        \item If $i \notin \widetilde{A}_i \cup \widetilde{B}_i$, set $X = X' \cap \xi'^{-1}(\widetilde{A}_i \cup \widetilde{B}_i)$ and $\xi = \xi'|_X$ (because then $u$ can map to $i'<i$).
        \item If $i \in \widetilde{A}_i \cup \widetilde{B}_i$, set $X = (X' \cap \xi'^{-1}(\widetilde{A}_i \cup \widetilde{B}_i)) \cup \{u\}$ and set $\xi$ such that $\xi|_{X \setminus \{u\}} = \xi'|_{X \setminus \{u\}}$ and $\xi(u) = i$.
    \end{itemize}
    \item Set $\textsf{Tab}[u,i,X,\xi] = \textsf{true}$.
\end{itemize}

As mentioned above, if $\textsf{Tab}$ contains a true entry such that the first coordinate of the indexing tuple is the last vertex of $G$, we get our ordered homomorphism $G\to H$. This completes the proof.
\end{proof}

Identifying all the classes of ordered graphs for which the \homo problem can be solved in polynomial time remains to be solved. We have shown in ~\cite{certik_matching_2025} that the \homo problem for ordered matchings, defined as ordered graphs where each vertex has exactly one incident edge, is \NP-complete.

The same article shows that deciding whether for a given ordered graph $G$ and ordered matching $M$ there exists an ordered homomorphism $G\to M$ is fixed-parameter tractable with respect to $|V(M)|$. We observe that while the ordered graphs $G$ and $H$ in Theorem ~\ref{thm:wone} can be complex, the ordered graph candidates for $H$, for which this problem is fixed-parameter tractable, are rather restricted in article ~\cite{certik_matching_2025}. Therefore, it remains to fill the gap between these classes of ordered graphs and categorize which of the classes of ordered graphs make the \parhomo problem \wone-hard and which of them make it fixed-parameter tractable. We shall not address this
question in the article at hand.


As mentioned above, determining whether \parhomo is \textsf{W[$i$]}-complete W[$i$]-complete for any [or for every?] $i$, also remains open.

%
%
%
\bibliographystyle{splncs04}
\bibliography{mybibliography}
\end{document}